\documentclass[a4paper,11pt]{article}

\usepackage[top=1in,bottom=1in,left=1in,right=1in]{geometry}
\usepackage{amsmath,amssymb,amsthm,mathrsfs}
\usepackage{microtype}
\usepackage{cite}
\usepackage{needspace}
\usepackage{hyperref}
\hypersetup{hidelinks,pdfauthor={Shangshuai Li},
  pdftitle={On equivalence between solution structures of the anti-self-dual Yang--Mills equations}}
\numberwithin{equation}{section}
\allowdisplaybreaks[2]
\theoremstyle{plain}
\newtheorem{theorem}{Theorem}[section]
\newtheorem{proposition}[theorem]{Proposition}

\newtheorem{corollary}[theorem]{Corollary}

\theoremstyle{definition}

\theoremstyle{remark}
\newtheorem{remark}[theorem]{Remark}

\newcommand{\Id}{I}
\newcommand{\Mat}{\operatorname{Mat}}
\newcommand{\GL}{\operatorname{GL}}

\newcommand{\SL}{\operatorname{SL}}

\newcommand{\tr}{\operatorname{tr}}
\newcommand{\diag}{\operatorname{diag}}

\title{%
	\fontsize{15}{20}\selectfont\bfseries 
	On equivalence between solution structures of the anti-self-dual Yang--Mills equations
}
\author{Shangshuai Li$^{1,2}$\thanks{Email: lishangshuai@nbu.edu.cn}
	\\[5pt]
	{\small $^1$ School of Mathematics and Statistics, Ningbo University, Ningbo 315211, China} \\
	{\small $^2$ Graduate School of Mathematics, Nagoya University, Nagoya 464-8602, Japan} 
}
\date{\today}

\begin{document}
\maketitle

\begin{abstract}
	We establish an exact correspondence between the quasi-Wronskian and quasi-Grammian constructions for the anti-self-dual Yang--Mills (ASDYM) equations with complex \(2\times2\) matrix-valued fields. Starting from vacuum wave functions associated with \(2M\) pairwise distinct nonzero constant spectral parameters, we restrict attention to a nonempty local domain on which the corresponding interpolation matrix is invertible. We show that both constructions are governed by the same uniquely solvable matrix-polynomial interpolation problem. In particular, the quasi-Wronskian Darboux polynomial \(T(\lambda)\) and the quasi-Grammian rational dressing matrix \(W(\lambda)\) satisfy
	\[
	T(\lambda)=q_\xi(\lambda)W(\lambda),
	\]
	where \(q_\xi(\lambda)\) is the monic scalar polynomial whose zeros are the prescribed poles of \(W(\lambda)\). This identity yields explicit relations between the corresponding Yang and potential matrices and shows that the two constructions generate identical ASDYM gauge potentials in the same gauge.
\end{abstract}



%\tableofcontents
%\newpage

\section{Introduction}

A characteristic feature of integrable systems is that the same family of exact solutions often admits several apparently different representations. A classical example is provided by the Kadomtsev--Petviashvili (KP) equation. Its dependent variable can be expressed in terms of a scalar $\tau$-function, under which the KP equation takes bilinear form. For $N$-soliton solutions, the $\tau$-function admits various determinant representations closely related to the Sato Grassmannian \cite{Sato1981,DateJimboKashiwaraMiwa1982,Kodama2017}. In particular, the Wronskian and Grammian constructions provide two classical determinant representations of KP solutions \cite{FreemanNimmo1983,Nakamura1989}. For KP line solitons, an explicit one-to-one correspondence between them has been established \cite{ChakravartyLewkowMaruno2010}.

Motivated by this correspondence in the KP theory, a natural question is whether an analogous exact correspondence exists for the anti-self-dual Yang--Mills (ASDYM) equations. The ASDYM equations constitute a matrix-valued integrable system of central importance in four-dimensional gauge theory and provide a unifying framework for many lower-dimensional integrable equations \cite{Ward1985,AblowitzChakravartyTakhtajan1993,MasonWoodhouse1996}. Their exact solutions have been constructed by a variety of methods, including instanton techniques, B\"acklund transformations, Riemann--Hilbert problem, bilinearization, and Darboux transformations \cite{BelavinPolyakovSchwarzTyupkin1975,AtiyahHitchinDrinfeldManin1978,UenoNakamura1982,SasaOhtaMatsukidaira1998,NimmoGilsonOhta2000}.

For our purpose, the most relevant constructions are the Wronskian-type solutions generated by Darboux transformations and the Grammian-type solutions arising from binary Darboux transformations or the Cauchy matrix approach \cite{NimmoGilsonOhta2000,LiQuYiZhang2022,LiQuZhang2023}. Since the relevant Darboux constructions involve noncommuting matrix-valued quantities, their closed-form solutions are naturally expressed in terms of quasideterminants and are commonly referred to as quasi-Wronskian and quasi-Grammian solutions \cite{GilsonHamanakaHuangNimmo2020,HamanakaHuang2020,LiHamanakaHuangZhang2025}.

Compared with the KP case, the connection between these two representations remains much less explicit. Previous studies have found agreement in their action densities, and more recently the Yang matrices of the general $N$-soliton solutions were shown to be asymptotically equivalent \cite{LiHamanakaHuangZhang2025,Huang2026}. A direct algebraic identification of the quasi-Wronskian and quasi-Grammian constructions remains to be established.

We establish an exact algebraic correspondence between the two
constructions for complex $2\times2$ ASDYM fields. The main result can be summarized as follows.

\paragraph{Main result.}
Let $\phi_i$ and $\psi_i$, $i=1,\ldots,M$, be
$\mathbb C^2$-valued vacuum wave functions associated with
$2M$ pairwise distinct nonzero constant spectral parameters
$\xi_i$ and $\gamma_i$, respectively.
Assume that the corresponding interpolation matrix $\mathcal V$, see \eqref{eq:mathcalV}, is invertible on a nonempty
local domain. Set
\[
q_\xi(\lambda)=\prod_{i=1}^{M}(\lambda-\xi_i),
\qquad
\Xi=\diag(\xi_1,\ldots,\xi_M).
\]
Let $T(\lambda)$ denote the quasi-Wronskian Darboux polynomial \cite{Cieslinski2009}
and $W(\lambda)$ the quasi-Grammian dressing matrix \cite{Takasaki2026}.
Using the subscripts $\mathrm W$ and $\mathrm G$ to distinguish
their associated determinant $\tau$-functions and Yang and
potential matrices, we obtain
\begin{subequations}
	\label{eq:intro-main-result}
	\begin{align}
		\tau_{\mathrm W}
		&=
		(-1)^{M(M-1)/2}
		\prod_{i,j=1}^{M}(\gamma_j-\xi_i)\,
		\tau_{\mathrm G},
		\label{eq:intro-tau-relation}
		\\
		T(\lambda)
		&=
		q_\xi(\lambda)W(\lambda), \qquad
		\mathcal J_{\mathrm W}
		=
		(-1)^M(\det\Xi)\mathcal J_{\mathrm G},
		\qquad
		\mathcal K_{\mathrm W}
		=
		\mathcal K_{\mathrm G}
		-(\tr\Xi)\Id_2.
		\label{eq:intro-field-relations}
	\end{align}
\end{subequations}
The determinant identity
\eqref{eq:intro-tau-relation} shows that the two determinant
representations have the same nonvanishing domain \cite[Eq.~(1.7)]{IshikawaOkadaTagawaZeng2006}.
On this common domain, the matrix identities \eqref{eq:intro-field-relations}
hold. Since the scalar factor relating the Yang matrices and
the additive shift relating the potential matrices are
independent of the spacetime coordinates, the two constructions
produce identical ASDYM gauge potentials in the gauge
$A_{\widetilde z}=A_{\widetilde w}=0$ and \eqref{eq:gauge-potential}.

\vspace{\baselineskip}

This paper is organized as follows. Section~\ref{sec-2} reviews the ASDYM equations, quasideterminants, and the spectral interpolation framework used throughout the paper. Section~\ref{sec-3} develops the quasi-Wronskian construction in terms of the Darboux matrix polynomial. Section~\ref{sec-4} introduces the quasi-Grammian construction through the Cauchy matrix approach and compares its dressing matrix with the Darboux matrix polynomial, thereby establishing the exact equivalence between the two constructions. Section~\ref{sec-5} concludes the paper with a summary of the main results and a brief discussion of possible further developments.

\section{Preliminaries}
\label{sec-2}

This section collects the basic ingredients needed for the comparison of the quasi-Wronskian and quasi-Grammian constructions. 

\subsection{The anti-self-dual Yang--Mills equations}

We first recall the formulation of the ASDYM equations used throughout this paper \cite{MasonWoodhouse1996}. Let \(\mathcal M\) be a four-dimensional complex manifold with local coordinates \((z,\widetilde z,w,\widetilde w)\), equipped with the flat metric
\[
ds^2=2\bigl(dz\,d\widetilde z-dw\,d\widetilde w\bigr).
\]
We take the gauge group to be \(G=\GL_2(\mathbb C)\), with Lie algebra \(\mathfrak g=\Mat_2(\mathbb C)\).
In the gauge
\begin{equation*}
A_{\widetilde z}=A_{\widetilde w}=0,
\end{equation*}
the ASDYM equations arise as the compatibility condition of the linear
system
\begin{equation}
	(\partial_z-\lambda\partial_{\widetilde w}+A_z)\Phi(\lambda)=0,
	\qquad
	(\partial_w-\lambda\partial_{\widetilde z}+A_w)\Phi(\lambda)=0,
	\label{eq:linear-system}
\end{equation}
where $\lambda$ is the spectral parameter.

The remaining gauge potentials admit two equivalent parametrizations,
\begin{equation}
	A_z
	=
	-(\partial_z\mathcal J)\mathcal J^{-1}
	=
	\partial_{\widetilde w}\mathcal K,
	\qquad
	A_w
	=
	-(\partial_w\mathcal J)\mathcal J^{-1}
	=
	\partial_{\widetilde z}\mathcal K.
	\label{eq:gauge-potential}
\end{equation}
These relations may be viewed as a Miura-type transformation between
the Yang matrix $\mathcal J$ and the potential matrix $\mathcal K$.
Substitution of the first parametrization into the compatibility
condition of \eqref{eq:linear-system} yields Yang's form of the ASDYM
equation \cite{Yang1977},
\begin{equation}
	\partial_{\widetilde z}
	\bigl((\partial_z\mathcal J)\mathcal J^{-1}\bigr)
	-
	\partial_{\widetilde w}
	\bigl((\partial_w\mathcal J)\mathcal J^{-1}\bigr)
	=0,
	\label{eq:J-ASDYM}
\end{equation}
whereas the second parametrization gives the potential form
\cite{Leznov1987},
\begin{equation}
	\partial_z\partial_{\widetilde z}\mathcal K
	-
	\partial_w\partial_{\widetilde w}\mathcal K
	+
	\bigl[
	\partial_{\widetilde w}\mathcal K,
	\partial_{\widetilde z}\mathcal K
	\bigr]
	=0.
	\label{eq:K-ASDYM}
\end{equation}
Thus, $\mathcal J$ and $\mathcal K$ provide two equivalent
parametrizations of the same ASDYM gauge field.

Accordingly, we distinguish between exact algebraic relations among
the solution matrices and equivalence at the level of the associated
ASDYM gauge field.

\begin{remark}
	For the comparison of different solution constructions, it is
important that these parametrizations are not uniquely normalized.
If $C\in\GL_2(\mathbb C)$ and
$C_0\in\Mat_2(\mathbb C)$ are constant matrices, then
\begin{equation*}
	\mathcal J\longmapsto\mathcal J C,
	\qquad
	\mathcal K\longmapsto\mathcal K+C_0
	\label{eq:JK-constant-freedom}
\end{equation*}
leave the gauge potentials in \eqref{eq:gauge-potential} unchanged.
Consequently, Yang matrices differing by a constant right factor, or
potential matrices differing by an additive constant matrix, determine
the same ASDYM gauge field. This normalization freedom will be used when comparing the different solution representations below.
\end{remark}

\subsection{Quasideterminants}

We next recall the quasideterminant notation used throughout the paper \cite{GelfandGelfandRetakh2005}.
Let
\begin{equation*}
	H=
	\begin{pmatrix}
		A&B\\
		C&D
	\end{pmatrix},
\end{equation*}
where the blocks have compatible sizes and $A$ is invertible. The
quasideterminant of $H$ expanded about the block $D$ is defined by
\begin{equation*}
	\left|
	\begin{array}{cc}
		A&B\\
		C&\boxed{D}
	\end{array}
	\right|
	:=
	D-CA^{-1}B.
	\label{eq:quasideterminant-definition}
\end{equation*}
Thus, in the block setting relevant here, a quasideterminant is precisely
the Schur complement of the complementary block $A$. The box specifies
the block with respect to which the quasideterminant is taken.

Quasideterminants provide a noncommutative analogue of determinant
ratios and are particularly suited to matrices whose entries are
themselves matrix-valued functions. In the constructions below, quasi-Wronskian and
quasi-Grammian solutions are naturally expressed in this notation.

\subsection{Spectral data and matrix polynomial interpolation}
\label{sec:spectral}

The purpose of this subsection is to recast the common spectral data of the quasi-Wronskian and quasi-Grammian constructions as a matrix-polynomial interpolation problem \cite{GohbergLancasterRodman1982,BallGohbergRodman1990}. 

For $M\geq1$, let
$
\{\xi_i\}_{i=1}^M,
$ and $
\{\gamma_i\}_{i=1}^M
$
be $2M$ pairwise distinct nonzero spectral parameters.
Associated with these spectral points, choose $\mathbb C^2$-valued
solutions of the vacuum linear system
\begin{equation}
	\partial_z\phi_i
	=
	\xi_i\partial_{\widetilde w}\phi_i,
	\qquad
	\partial_w\phi_i
	=
	\xi_i\partial_{\widetilde z}\phi_i,
	\qquad
	\partial_z\psi_i
	=
	\gamma_i\partial_{\widetilde w}\psi_i,
	\qquad
	\partial_w\psi_i
	=
	\gamma_i\partial_{\widetilde z}\psi_i,
	\label{eq:dispersion-relation}
\end{equation}
where $\phi_i$ corresponds to the spectral point $\lambda=\xi_i$, whereas $\psi_i$ corresponds to $\lambda=\gamma_i$. Equivalently,
\begin{equation*}
	\phi_i
	=
	\begin{pmatrix}
		\phi_i^{(1)}
		(\xi_i z+\widetilde w,\,
		\xi_i w+\widetilde z)
		\\[1mm]
		\phi_i^{(2)}
		(\xi_i z+\widetilde w,\,
		\xi_i w+\widetilde z)
	\end{pmatrix},
	\qquad
	\psi_i
	=
	\begin{pmatrix}
		\psi_i^{(1)}
		(\gamma_i z+\widetilde w,\,
		\gamma_i w+\widetilde z)
		\\[1mm]
		\psi_i^{(2)}
		(\gamma_i z+\widetilde w,\,
		\gamma_i w+\widetilde z)
	\end{pmatrix}.
	\label{eq:vacuum-wave-functions}
\end{equation*}
These spectral points together with the corresponding vacuum wave functions form the common input data for the quasi-Wronskian and quasi-Grammian constructions.

The pairs $(\xi_i,\phi_i)$ and $(\gamma_i,\psi_i)$ therefore provide right spectral interpolation data. We seek a monic $2\times2$ matrix polynomial satisfying the corresponding homogeneous right interpolation conditions
\begin{equation}
	P(\xi_i)\phi_i=0,
	\qquad
	P(\gamma_i)\psi_i=0,
	\qquad
	i=1,\ldots,M.
	\label{eq:spectral-interpolation-conditions}
\end{equation}
Accordingly, let
\begin{equation*}
	P(\lambda)
	=
	\lambda^M\Id_2
	+
	\sum_{k=0}^{M-1}P_k\lambda^k.
\end{equation*}
Writing each coefficient matrix explicitly as
\begin{equation*}
	P_k
	=
	\begin{pmatrix}
		p_{k,11} & p_{k,12}\\
		p_{k,21} & p_{k,22}
	\end{pmatrix},
	\qquad
	k=0,\ldots,M-1,
\end{equation*}
the $a$th component of each condition in
\eqref{eq:spectral-interpolation-conditions}, for a fixed row index $a=1,2$, takes the form
\begin{align*}
	\sum_{k=0}^{M-1}
	\xi_i^k
	\bigl(
	p_{k,a1}\phi_i^{(1)}
	+
	p_{k,a2}\phi_i^{(2)}
	\bigr)
	&=
	-\xi_i^M\phi_i^{(a)},
	\\
	\sum_{k=0}^{M-1}
	\gamma_i^k
	\bigl(
	p_{k,a1}\psi_i^{(1)}
	+
	p_{k,a2}\psi_i^{(2)}
	\bigr)
	&=
	-\gamma_i^M\psi_i^{(a)}.
\end{align*}
Collecting these $2M$ scalar equations gives the linear system
\begin{equation}
	\mathcal V
	\begin{pmatrix}
		p_{0,a1}\\
		\vdots\\
		p_{M-1,a1}\\
		p_{0,a2}\\
		\vdots\\
		p_{M-1,a2}
	\end{pmatrix}
	=
	-
	\begin{pmatrix}
		\xi_1^M\phi_1^{(a)}\\
		\vdots\\
		\xi_M^M\phi_M^{(a)}\\
		\gamma_1^M\psi_1^{(a)}\\
		\vdots\\
		\gamma_M^M\psi_M^{(a)}
	\end{pmatrix},
	\label{eq:interpolation-linear-system}
\end{equation}
where $\mathcal V$ is the corresponding interpolation matrix:
\begin{equation}
	\mathcal V
	=
	\begin{pmatrix}
		\phi_1^{(1)}
		&
		\xi_1\phi_1^{(1)}
		&
		\cdots
		&
		\xi_1^{M-1}\phi_1^{(1)}
		&
		\phi_1^{(2)}
		&
		\xi_1\phi_1^{(2)}
		&
		\cdots
		&
		\xi_1^{M-1}\phi_1^{(2)}
		\\
		\vdots
		&
		\vdots
		&
		&
		\vdots
		&
		\vdots
		&
		\vdots
		&
		&
		\vdots
		\\
		\phi_M^{(1)}
		&
		\xi_M\phi_M^{(1)}
		&
		\cdots
		&
		\xi_M^{M-1}\phi_M^{(1)}
		&
		\phi_M^{(2)}
		&
		\xi_M\phi_M^{(2)}
		&
		\cdots
		&
		\xi_M^{M-1}\phi_M^{(2)}
		\\[1mm]
		\psi_1^{(1)}
		&
		\gamma_1\psi_1^{(1)}
		&
		\cdots
		&
		\gamma_1^{M-1}\psi_1^{(1)}
		&
		\psi_1^{(2)}
		&
		\gamma_1\psi_1^{(2)}
		&
		\cdots
		&
		\gamma_1^{M-1}\psi_1^{(2)}
		\\
		\vdots
		&
		\vdots
		&
		&
		\vdots
		&
		\vdots
		&
		\vdots
		&
		&
		\vdots
		\\
		\psi_M^{(1)}
		&
		\gamma_M\psi_M^{(1)}
		&
		\cdots
		&
		\gamma_M^{M-1}\psi_M^{(1)}
		&
		\psi_M^{(2)}
		&
		\gamma_M\psi_M^{(2)}
		&
		\cdots
		&
		\gamma_M^{M-1}\psi_M^{(2)}
		\label{eq:mathcalV}
	\end{pmatrix}.
\end{equation}
Throughout this paper, we work locally on a domain where $\mathcal V$ is invertible. Consequently, the linear systems \eqref{eq:interpolation-linear-system}, for $a=1,2$, uniquely determine the coefficient matrices $P_k$ and hence the monic interpolation polynomial $P(\lambda)$.

\begin{remark}
	The matrix $\mathcal V$ admits a representation of the double-Wronskian type. Set
	\begin{equation*}
		f^{(a)}=(\phi_1^{(a)},\ldots,\phi_M^{(a)},\psi_1^{(a)},\ldots,\psi_M^{(a)})^T,
		\qquad
		\chi=\diag(\xi_1,\ldots,\xi_M,\gamma_1,\ldots,\gamma_M),
	\end{equation*}
	then
	\begin{equation*}
		\det\mathcal V=
		\det\left(f^{(1)},\chi f^{(1)},\ldots,\chi^{M-1} f^{(1)}; f^{(2)},\chi f^{(2)},\ldots,\chi^{M-1} f^{(2)}\right).
	\end{equation*}
	This is precisely of the double-Wronskian type appearing in the determinant construction of \cite{KakeiIkedaTakasaki2002}, where double-Wronskian solutions are constructed for a $(2+1)$-dimensional integrable hierarchy closely related to the self-dual Yang--Mills equations.
\end{remark}

\begin{proposition}
	\label{prop:determinant}
	Define the
	monic scalar polynomials
	\begin{equation}
		q_\xi(\lambda)
		=
		\prod_{i=1}^M(\lambda-\xi_i),
		\qquad
		q_\gamma(\lambda)
		=
		\prod_{i=1}^M(\lambda-\gamma_i),
		\qquad
		p(\lambda)
		=
		\prod_{i=1}^M
		(\lambda-\xi_i)(\lambda-\gamma_i).
		\label{eq:spectral-polynomials}
	\end{equation}
	Then the determinant of $P(\lambda)$ satisfies
	\begin{equation}
		\det P(\lambda)
		=
		p(\lambda)
		=
		q_\xi(\lambda)q_\gamma(\lambda).
		\label{eq:det-interpolation-polynomial}
	\end{equation}
\end{proposition}

\begin{proof}
	Since $\mathcal V$ is invertible, none of the vectors appearing in the
	interpolation data can vanish. Therefore
	\eqref{eq:spectral-interpolation-conditions} implies
	\[
	\det P(\xi_i)=0,
	\qquad
	\det P(\gamma_i)=0.
	\]
	Hence the $2M$ pairwise distinct spectral points
	$
	\xi_1,\ldots,\xi_M,\gamma_1,\ldots,\gamma_M
	$
	are zeros of $\det P(\lambda)$. On the other hand, since
	$P(\lambda)$ is monic of degree $M$, $\det P(\lambda)$ is a monic
	scalar polynomial of degree $2M$. It follows that
	\[
	\det P(\lambda)
	=
	\prod_{i=1}^M
	(\lambda-\xi_i)(\lambda-\gamma_i)
	=
	p(\lambda),
	\]
	which proves \eqref{eq:det-interpolation-polynomial}.
\end{proof}

\section{Quasi-Wronskian construction}
\label{sec-3}

The quasi-Wronskian solutions of the ASDYM equations arise from Darboux transformations of the associated linear system and admit compact quasideterminant representations of Wronskian type \cite{GilsonHamanakaHuangNimmo2020}. 
In this section, we encode the same spectral data in a monic Darboux matrix polynomial characterized by the interpolation conditions \eqref{eq:spectral-interpolation-conditions}. This polynomial determines the quasi-Wronskian solution of the ASDYM equations.

\subsection{Darboux matrix polynomial}

For $i=1,\ldots,M$, define
\begin{equation}
	\Theta_i=\left(\phi_i,\psi_i\right)=
	\begin{pmatrix}
		\phi_i^{(1)} & \psi_i^{(1)} \\
		\phi_i^{(2)} & \psi_i^{(2)}
	\end{pmatrix},
	\qquad
	\Lambda_i=\diag(\xi_i,\gamma_i),
	\qquad
	\Pi_k=\left(\Theta_1\Lambda_1^k,\ldots,\Theta_M\Lambda_M^k\right).
\end{equation}
The dispersion relations \eqref{eq:dispersion-relation} can then be written as
\begin{equation}
	\partial_z\Theta_i
	=(\partial_{\widetilde w}\Theta_i)\Lambda_i,
	\qquad
	\partial_w\Theta_i
	=(\partial_{\widetilde z}\Theta_i)\Lambda_i.
	\label{eq:qw-dispersion}
\end{equation}

Introduce the $2M\times2M$ matrix with $2\times2$ blocks
\begin{equation}
	\Sigma:=
	\begin{pmatrix}
		\Theta_1 & \cdots & \Theta_M\\
		\Theta_1\Lambda_1 & \cdots & \Theta_M\Lambda_M\\
		\vdots & & \vdots\\
		\Theta_1\Lambda_1^{M-1}
		& \cdots & \Theta_M\Lambda_M^{M-1}
	\end{pmatrix}
	=\begin{pmatrix}
		\Pi_0 \\
		\Pi_1 \\
		\vdots \\
		\Pi_{M-1}
	\end{pmatrix}.
	\label{eq:qw-block-matrix}
\end{equation}
Let \(V_\xi=(\xi_j^r)_{r=0,\ldots,M-1}^{j=1,\ldots,M}\) and \(V_\gamma=(\gamma_j^r)_{r=0,\ldots,M-1}^{j=1,\ldots,M}\) be the Vandermonde matrices.
Reordering the two components in the rows and columns of $\Sigma$ gives
permutations of the same parity. Therefore
\begin{equation}
	\det\Sigma
	=
	\det
	\begin{pmatrix}
		V_\xi\diag(\phi_1^{(1)},\ldots,\phi_M^{(1)})
		&
		V_\gamma\diag(\psi_1^{(1)},\ldots,\psi_M^{(1)})
		\\[2mm]
		V_\xi\diag(\phi_1^{(2)},\ldots,\phi_M^{(2)})
		&
		V_\gamma\diag(\psi_1^{(2)},\ldots,\psi_M^{(2)})
	\end{pmatrix}.
	\label{eq:step-det-sigma}
\end{equation}
The matrix on the right-hand side is the transpose of $\mathcal V$. Hence
$\det\Sigma=\det\mathcal V\neq0$, so $\Sigma$ is invertible on the
local domain under consideration.

Following the Darboux transformation and its
quasi-Wronskian representation
\cite{NimmoGilsonOhta2000,GilsonHamanakaHuangNimmo2020},
we define the associated monic Darboux matrix polynomial of degree $M$ by \cite{Cieslinski2009}
\begin{equation}
		T(\lambda):=
		\left|
		\begin{array}{cccc}
			\Theta_1 & \cdots & \Theta_M & \Id_2\\
			\Theta_1\Lambda_1 & \cdots & \Theta_M\Lambda_M
			& \lambda\Id_2\\
			\vdots & & \vdots & \vdots\\
			\Theta_1\Lambda_1^{M-1}
			& \cdots & \Theta_M\Lambda_M^{M-1}
			& \lambda^{M-1}\Id_2\\
			\Theta_1\Lambda_1^M
			& \cdots & \Theta_M\Lambda_M^M
			& \boxed{\lambda^M\Id_2}
		\end{array}
		\right|
		\quad
		=\left|
		\begin{array}{cc}
			\Sigma & E(\lambda)\\
			\Pi_M & \boxed{\lambda^M\Id_2}
		\end{array}
		\right|
		=\lambda^M\Id_2+\sum_{k=0}^{M-1}T_k\lambda^k.
	\label{eq:qw-Darboux-polynomial}
\end{equation}
Here $E(\lambda)=(\Id_2,\lambda\Id_2,\ldots,
\lambda^{M-1}\Id_2)^{T}$. Expanding the quasideterminant as a
Schur complement shows that the coefficient matrices satisfy
\begin{equation*}
	\Pi_M+	(T_0,\ldots,T_{M-1})\Sigma=0.
\end{equation*}
Reading this identity block column by block column gives
\begin{equation}
	\Theta_i\Lambda_i^M
	+
	\sum_{k=0}^{M-1}T_k\Theta_i\Lambda_i^k
	=0,
	\qquad i=1,\ldots,M.
	\label{eq:qw-annihilation}
\end{equation}
Its two columns are
\begin{equation*}
	\left(
	\xi_i^M\phi_i,
	\gamma_i^M\psi_i
	\right)
	+
	\sum_{k=0}^{M-1}
	\left(
	T_k\xi_i^k\phi_i,
	T_k\gamma_i^k\psi_i
	\right)
	=0.
\end{equation*}
Consequently, $T(\lambda)$ satisfies the interpolation conditions (cf. \eqref{eq:spectral-interpolation-conditions})
\begin{equation*}
	T(\xi_i)\phi_i=0,
	\qquad
	T(\gamma_i)\psi_i=0.
\end{equation*}
Since $T(\lambda)$ is monic and of degree $M$, the discussion in Sec.~\ref{sec:spectral}
implies
$
T(\lambda)=P(\lambda).
$
Thus the Darboux matrix polynomial is precisely the unique matrix
polynomial determined by the interpolation conditions \eqref{eq:spectral-interpolation-conditions}.

\subsection{Corresponding ASDYM fields}

Since $T(\lambda)=P(\lambda)$ and $T_0=T(0)$,
equation \eqref{eq:det-interpolation-polynomial} gives
\begin{align}
	\det T_0=\det T(0)=\det P(0)=p(0)=\prod_{i=1}^M
	\xi_i\gamma_i\neq0.
	\label{eq:det-T0}
\end{align}
Thus $T_0$ is invertible. The evolution of the Darboux matrix polynomial is
described by the following theorem.

\begin{theorem}
	\label{thm:qw-sato-wilson}
	The Darboux matrix polynomial
	$T(\lambda)$ satisfies
	\begin{subequations}
		\label{eq:qw-sato-wilson}
		\begin{align}
			\left(
			\partial_z-\lambda\partial_{\widetilde w}
			+\partial_{\widetilde w}T_{M-1}
			\right)T(\lambda)
			&=0,
			\label{eq:qw-sato-wilson-z}\\
			\left(
			\partial_w-\lambda\partial_{\widetilde z}
			+\partial_{\widetilde z}T_{M-1}
			\right)T(\lambda)
			&=0.
			\label{eq:qw-sato-wilson-w}
		\end{align}
	\end{subequations}
	Moreover, since $T_0$ is invertible by \eqref{eq:det-T0}, we have
	\begin{equation}
		-(\partial_zT_0)T_0^{-1}
		=\partial_{\widetilde w}T_{M-1},
		\qquad
		-(\partial_wT_0)T_0^{-1}
		=\partial_{\widetilde z}T_{M-1}.
		\label{eq:wronskian-JK}
	\end{equation}
\end{theorem}

\begin{proof}
	Differentiate \eqref{eq:qw-annihilation} with respect to $z$ and
	subtract its $\widetilde w$-derivative multiplied on the right by
	$\Lambda_i$. By the dispersion relations
	\eqref{eq:qw-dispersion}, all terms involving derivatives of
	$\Theta_i$ cancel, leaving
	\begin{equation}
		\sum_{k=0}^{M-1}
		\left[
		(\partial_zT_k)\Theta_i\Lambda_i^k
		-
		(\partial_{\widetilde w}T_k)
		\Theta_i\Lambda_i^{k+1}
		\right]
		=0.
		\label{eq:step-1}
	\end{equation}
	Set $T_{-1}:=0$ and define
	\[
	R_k:=
	\partial_zT_k
	-
	\partial_{\widetilde w}T_{k-1}
	+
	(\partial_{\widetilde w}T_{M-1})T_k,
	\qquad k=0,\ldots,M-1.
	\]
	Using \eqref{eq:qw-annihilation} to eliminate
	$\Theta_i\Lambda_i^M$, equation \eqref{eq:step-1} becomes
	\[
	\sum_{k=0}^{M-1}R_k\Theta_i\Lambda_i^k=0,
	\qquad i=1,\ldots,M.
	\]
	Collecting these identities for $i=1,\ldots,M$ yields
	$
	(R_0,\ldots,R_{M-1})\Sigma=0.
	$
	Since $\Sigma$ is invertible, it follows that $R_k=0$ for all
	$k=0,\ldots,M-1$. Therefore
	\[
	\left(
	\partial_z-\lambda\partial_{\widetilde w}
	+\partial_{\widetilde w}T_{M-1}
	\right)T(\lambda)
	=
	\sum_{k=0}^{M-1}R_k\lambda^k
	=0,
	\]
	which proves \eqref{eq:qw-sato-wilson-z}. The same argument, with
	$(z,\widetilde w)$ replaced by $(w,\widetilde z)$, proves
	\eqref{eq:qw-sato-wilson-w}.
	Finally, evaluating \eqref{eq:qw-sato-wilson} at $\lambda=0$ and
	multiplying on the right by $T_0^{-1}$ gives
	\eqref{eq:wronskian-JK}.
\end{proof}

Comparison of \eqref{eq:wronskian-JK} with the Miura relations
\eqref{eq:gauge-potential} identifies the Yang and potential matrices
of the quasi-Wronskian construction as
\begin{equation}
	\mathcal J_{\mathrm W}:=T_0,
	\qquad
	\mathcal K_{\mathrm W}:=T_{M-1}.
	\label{eq:qw-JK-fields}
\end{equation}
They recover the ASDYM fields through \eqref{eq:gauge-potential}.
The associated scalar $\tau$-function is
\begin{equation}
	\tau_{\mathrm W}:=\det\Sigma.
	\label{eq:qw-tau}
\end{equation}

\begin{remark}
	Extracting the coefficients of $\lambda^0$ and $\lambda^{M-1}$
	from \eqref{eq:qw-Darboux-polynomial} gives
	\begin{align*}
		T_0
		=
		\left|
		\begin{array}{cccc}
			\Theta_1 & \cdots & \Theta_M & \Id_2\\
			\Theta_1\Lambda_1 & \cdots & \Theta_M\Lambda_M & 0\\
			\vdots & & \vdots & \vdots\\
			\Theta_1\Lambda_1^{M-1}
			& \cdots & \Theta_M\Lambda_M^{M-1} & 0\\
			\Theta_1\Lambda_1^M
			& \cdots & \Theta_M\Lambda_M^M & \boxed{0}
		\end{array}
		\right|, \qquad
		T_{M-1}
		=
		\left|
		\begin{array}{cccc}
			\Theta_1 & \cdots & \Theta_M & 0\\
			\Theta_1\Lambda_1 & \cdots & \Theta_M\Lambda_M & 0\\
			\vdots & & \vdots & \vdots\\
			\Theta_1\Lambda_1^{M-1}
			& \cdots & \Theta_M\Lambda_M^{M-1} & \Id_2\\
			\Theta_1\Lambda_1^M
			& \cdots & \Theta_M\Lambda_M^M & \boxed{0}
		\end{array}
		\right|.
	\end{align*}
	These two coefficients recover exactly the quasi-Wronskian 
	used in the quasideterminant representations of Yang's $J$-matrix
	and the potential $K$-matrix, respectively
	\cite{GilsonHamanakaHuangNimmo2020,HamanakaHuang2022}.
	In particular, the position of the identity block in the last column
	distinguishes the two representations.
\end{remark}

\section{Quasi-Grammian construction}
\label{sec-4}

We now turn to the Cauchy matrix formulation of the ASDYM equations developed in \cite{LiQuYiZhang2022,LiHamanakaHuangZhang2025} and use it to construct quasi-Grammian solutions. To facilitate a direct comparison with the quasi-Wronskian construction, we adopt the same spectral data and vacuum wave functions introduced in Sec.~\ref{sec:spectral}.

\subsection{Cauchy matrix approach}

Using the spectral data introduced above, we now formulate the corresponding Cauchy matrix construction. Consider the Sylvester equation
\begin{equation}
	\Omega\Gamma-\Xi\Omega=\rho\theta,
	\label{eq:cauchy-sylvester-equation}
\end{equation}
with diagonal spectral matrices
\begin{equation}
	\Xi=\diag(\xi_1,\ldots,\xi_M),
	\qquad
	\Gamma=\diag(\gamma_1,\ldots,\gamma_M),
	\label{eq:cauchy-diagonal-spectra}
\end{equation}
and the associated vacuum wave-function matrices
\begin{equation}
	\rho=
	\begin{pmatrix}
		-\phi_1^{(2)} & \phi_1^{(1)}\\
		\vdots & \vdots\\
		-\phi_M^{(2)} & \phi_M^{(1)}
	\end{pmatrix},
	\qquad
	\theta=
	\begin{pmatrix}
		\psi_1^{(1)} & \cdots & \psi_M^{(1)}\\[5pt]
		\psi_1^{(2)} & \cdots & \psi_M^{(2)}
	\end{pmatrix}.
	\label{eq:gram-wave-data}
\end{equation}
By \eqref{eq:dispersion-relation}, these matrices satisfy the dispersion relations
\begin{equation}
	\begin{aligned}
		\partial_z\rho=\Xi(\partial_{\widetilde{w}}\rho),
		\qquad
		\partial_z\theta=(\partial_{\widetilde{w}}\theta) \Gamma, \\
		\partial_w\rho=\Xi(\partial_{\widetilde{z}}\rho),
		\qquad
		\partial_w\theta=(\partial_{\widetilde{z}}\theta) \Gamma.
	\end{aligned}
	\label{eq:dispersion}
\end{equation}
Since the spectra of $\Xi$ and $\Gamma$ are disjoint, the Sylvester
equation \eqref{eq:cauchy-sylvester-equation} admits a unique solution
$\Omega=(\Omega_{ij})_{i,j=1}^{M}$. Its entries are given explicitly by
\begin{equation}
	\Omega_{ij}
	=
	\frac{
		\phi_i^{(1)}\psi_j^{(2)}-\phi_i^{(2)}\psi_j^{(1)}
	}{
		\gamma_j-\xi_i
	},
	\qquad i,j=1,\ldots,M.
	\label{eq:gram-cauchy-kernel}
\end{equation}
Thus, $\Omega$ is a Cauchy-type matrix. The associated scalar $\tau$-function is
\begin{equation}
	\tau_{\mathrm{G}}=\det\Omega
	=
	\det\left[
	\frac{
		\phi_i^{(1)}\psi_j^{(2)}-\phi_i^{(2)}\psi_j^{(1)}
	}{
		\gamma_j-\xi_i
	}
	\right]_{i,j=1}^{M}.
	\label{eq:gram-tau}
\end{equation}
It serves as the common scalar denominator of the quasi-Grammian expressions appearing below. In what follows, we work locally on a domain where
$\tau_{\mathrm G}\neq0$, so that $\Omega$ is invertible.

Recent work has clarified the relation between the Cauchy matrix approach and the Sato-theoretic dressing formulation of integrable hierarchies, including the ASDYM hierarchy \cite{LiZhang2026,Takasaki2026}. In this framework, we introduce the rational matrix function
\begin{equation}
	W(\lambda)
	:=
	\left|
	\begin{array}{cc}
		\Omega & (\lambda\Id_M-\Xi)^{-1}\rho\\
		\theta & \boxed{\Id_2}
	\end{array}
	\right|
	=
	\Id_2-\theta\Omega^{-1}
	(\lambda\Id_M-\Xi)^{-1}\rho,
	\label{eq:W-rational}
\end{equation}
which plays the role of the corresponding Sato--Wilson dressing matrix.

The Yang and potential matrices are recovered from the value of $W(\lambda)$ at $\lambda=0$ and from its expansion at infinity:
\begin{equation}
	\mathcal J_{\mathrm{G}}=W(0),
	\qquad
	W(\lambda)
	=
	\Id_2+\lambda^{-1}\mathcal K_{\mathrm{G}}
	+O(\lambda^{-2}),
	\qquad
	\lambda\to\infty.
	\label{eq:gram-dressing-recovery}
\end{equation}
Equivalently, they admit the quasideterminant representations
\begin{equation*}
	\begin{aligned}
		\mathcal J_{\mathrm{G}}
		&=
		\left|
		\begin{array}{cc}
			\Omega & -\Xi^{-1}\rho\\
			\theta & \boxed{\Id_2}
		\end{array}
		\right|,
		\qquad
		\mathcal K_{\mathrm{G}}
		=
		\left|
		\begin{array}{cc}
			\Omega & \rho\\
			\theta & \boxed{0}
		\end{array}
		\right|.
	\end{aligned}
	\label{eq:gram-JK}
\end{equation*}
We show below that these matrices satisfy the Miura relations \eqref{eq:gauge-potential}, thereby confirming $\mathcal J_{\mathrm G}$ and $\mathcal K_{\mathrm G}$ as the Yang and potential matrices of the quasi-Grammian ASDYM solution.

\subsection{Equivalence with the quasi-Wronskian construction}

We now compare the quasi-Grammian construction with the quasi-Wronskian formulation developed in the previous section. 

\begin{theorem}
	The quasi-Wronskian and quasi-Grammian $\tau$-functions are related by (cf. \cite[Eq.~(1.7)]{IshikawaOkadaTagawaZeng2006})
	\begin{equation}
		\tau_{\mathrm{W}}
		=
		(-1)^{M(M-1)/2}
		\prod_{i,j=1}^M(\gamma_j-\xi_i)
		\tau_{\mathrm{G}}.
	\end{equation}
	In particular, since $\gamma_j\neq\xi_i$ for all $i,j$,
	$\tau_{\mathrm W}$ and $\tau_{\mathrm G}$ vanish simultaneously and
	therefore have the same nonvanishing domain.
\end{theorem}

\begin{proof}
	Recall \eqref{eq:step-det-sigma}.
	Since the $\xi_i$ are pairwise distinct, $V_\xi$ is invertible.
	Factoring $\diag(V_\xi,V_\xi)$ from the left gives
	\begin{align}
		\det\Sigma
		&=
		(\det V_\xi)^2
		\det
		\begin{pmatrix}
			\diag(\phi_1^{(1)},\ldots,\phi_M^{(1)})
			&
			V_\xi^{-1}V_\gamma
			\diag(\psi_1^{(1)},\ldots,\psi_M^{(1)})
			\\
			\diag(\phi_1^{(2)},\ldots,\phi_M^{(2)})
			&
			V_\xi^{-1}V_\gamma
			\diag(\psi_1^{(2)},\ldots,\psi_M^{(2)})
		\end{pmatrix}
		\nonumber\\
		&=
		(\det V_\xi)^2
		\det\left[
		(V_\xi^{-1}V_\gamma)_{ij}
		\left(
		\phi_i^{(1)}\psi_j^{(2)}
		-\phi_i^{(2)}\psi_j^{(1)}
		\right)
		\right]_{i,j=1}^M,
		\label{eq:step-det-reduction}
	\end{align}
	where the second equality follows from the Laplace expansion with
	respect to the first $M$ columns.
	By the Lagrange interpolation formula,
	\begin{align}
		(V_\xi^{-1}V_\gamma)_{ij}
		=
		\frac{\displaystyle
			\prod_{\substack{k=1\\k\neq i}}^M
			(\gamma_j-\xi_k)}
		{\displaystyle
			\prod_{\substack{k=1\\k\neq i}}^M
			(\xi_i-\xi_k)}
		=
		\frac{\displaystyle
			\prod_{k=1}^M(\gamma_j-\xi_k)}
		{\displaystyle
			(\gamma_j-\xi_i)
			\prod_{\substack{k=1\\k\neq i}}^M
			(\xi_i-\xi_k)}.
		\label{eq:step-lagrange-matrix}
	\end{align}
	Substituting \eqref{eq:step-lagrange-matrix} into
	\eqref{eq:step-det-reduction} and extracting the corresponding row
	and column factors yields
	\begin{align}
		\det\Sigma
		&=
		(\det V_\xi)^2
		\frac{\displaystyle
			\prod_{j=1}^M\prod_{k=1}^M
			(\gamma_j-\xi_k)}
		{\displaystyle
			\prod_{i=1}^M
			\prod_{\substack{k=1\\k\neq i}}^M
			(\xi_i-\xi_k)}
		\det\left[
		\frac{
			\phi_i^{(1)}\psi_j^{(2)}
			-\phi_i^{(2)}\psi_j^{(1)}
		}{
			\gamma_j-\xi_i
		}
		\right]_{i,j=1}^M.
		\label{eq:step-det-cauchy}
	\end{align}
	Using
	\begin{align*}
		\prod_{i=1}^M
		\prod_{\substack{k=1\\k\neq i}}^M
		(\xi_i-\xi_k)
		=
		(-1)^{M(M-1)/2}
		\prod_{1\leq i<k\leq M}
		(\xi_k-\xi_i)^2
		=
		(-1)^{M(M-1)/2}(\det V_\xi)^2,
	\end{align*}
	equation \eqref{eq:step-det-cauchy} reduces to
	\begin{align*}
		\det\Sigma
		&=
		(-1)^{M(M-1)/2}
		\prod_{i,j=1}^M(\gamma_j-\xi_i)
		\det\left[
		\frac{
			\phi_i^{(1)}\psi_j^{(2)}
			-\phi_i^{(2)}\psi_j^{(1)}
		}{
			\gamma_j-\xi_i
		}
		\right]_{i,j=1}^M
		\nonumber\\
		&=
		(-1)^{M(M-1)/2}
		\prod_{i,j=1}^M(\gamma_j-\xi_i)
		\det\Omega.
	\end{align*}
	Since $\tau_{\mathrm W}=\det\Sigma$ and
	$\tau_{\mathrm G}=\det\Omega$, the claimed identity follows.
\end{proof}

The preceding result shows that $\tau_{\mathrm W}$ and $\tau_{\mathrm G}$ have the same nonvanishing domain. Since $\tau_{\mathrm W}=\det\mathcal V$, the condition $\tau_{\mathrm G}\neq0$ implies that $\mathcal V$ is invertible. Hence, on this common domain, the monic interpolation polynomial $P(\lambda)$ is uniquely determined by the conditions \eqref{eq:spectral-interpolation-conditions}.

This uniqueness gives a direct means of comparing the quasi-Grammian dressing matrix with the Darboux matrix polynomial. The resulting relation is the following.

\begin{theorem}
	\label{thm:qw-gram-dressing}
	The Darboux matrix polynomial \eqref{eq:qw-Darboux-polynomial}
	and the dressing matrix \eqref{eq:W-rational}
	satisfy
	\begin{equation}
		T(\lambda)
		=
		q_\xi(\lambda)W(\lambda),
		\qquad
		q_\xi(\lambda)
		=
		\prod_{i=1}^M(\lambda-\xi_i).
		\label{eq:qw-gram-dressing}
	\end{equation}
\end{theorem}

\begin{proof}
	For convenience, set
	\[
	\widehat P(\lambda)
	:=
	q_\xi(\lambda)W(\lambda).
	\]
	Since multiplication by $q_\xi(\lambda)$ cancels all possible finite poles of $W(\lambda)$, the matrix $\widehat P(\lambda)$ is polynomial in $\lambda$. Together with the normalization
	$W(\lambda)=\Id_2+O(\lambda^{-1})$ as $\lambda\to\infty$,
	this shows that $\widehat P(\lambda)$ is a monic $2\times2$ matrix polynomial of degree $M$.
	
	It remains to show that $\widehat P(\lambda)$ satisfies the interpolation conditions \eqref{eq:spectral-interpolation-conditions}. From \eqref{eq:W-rational}, we have
	\[
	\widehat P(\lambda)
	=
	\left(\prod_{j=1}^M(\lambda-\xi_j)\right)\Id_2
	-
	\theta\Omega^{-1}
	\left(\prod_{j=1}^M(\lambda-\xi_j)\right)
	(\lambda\Id_M-\Xi)^{-1}\rho.
	\]
	For $\lambda=\xi_i$, using the explicit form of $\rho$ in \eqref{eq:gram-wave-data}, we obtain
	\begin{equation*}
		\begin{aligned}
			\widehat P(\xi_i)
			\phi_i
			=
			-\theta\Omega^{-1}
			\left.
			\begin{pmatrix}
				\displaystyle\prod_{j\neq1}(\lambda-\xi_j)
				&&\\
				&\ddots&\\
				&&\displaystyle\prod_{j\neq M}(\lambda-\xi_j)
			\end{pmatrix}
			\begin{pmatrix}
				\phi_1^{(1)}\phi_i^{(2)}-\phi_1^{(2)}\phi_i^{(1)}\\
				\vdots\\
				\phi_i^{(1)}\phi_i^{(2)}-\phi_i^{(2)}\phi_i^{(1)}\\
				\vdots\\
				\phi_M^{(1)}\phi_i^{(2)}-\phi_M^{(2)}\phi_i^{(1)}
			\end{pmatrix}
			\right|_{\lambda=\xi_i}
			=0.
		\end{aligned}
	\end{equation*}
	Indeed, every component with index different from $i$
	contains the factor $\lambda-\xi_i$, whereas the $i$th
	component vanishes identically since
	$\phi_i^{(1)}\phi_i^{(2)}-\phi_i^{(2)}\phi_i^{(1)}=0$.
	Hence
	$\widehat P(\xi_i)\phi_i=0$, which establishes the first set of
	conditions in \eqref{eq:spectral-interpolation-conditions}.
	
	For the second set of conditions, the Sylvester equation
	\eqref{eq:cauchy-sylvester-equation} may be rewritten as
	\[
	(\lambda\Id_M-\Xi)\Omega-\rho\theta
	=
	\Omega(\lambda\Id_M-\Gamma).
	\]
	Therefore, for $\lambda\notin\{\xi_1,\ldots,\xi_M\}$,
	\begin{equation}
		\begin{aligned}
			W(\lambda)\theta
			&=
			\theta-\theta\Omega^{-1}
			(\lambda\Id_M-\Xi)^{-1}\rho\theta
			\\
			&=
			\theta-\theta\Omega^{-1}
			(\lambda\Id_M-\Xi)^{-1}
			\bigl[
			(\lambda\Id_M-\Xi)\Omega
			-\Omega(\lambda\Id_M-\Gamma)
			\bigr]
			\\
			&=
			\theta\Omega^{-1}
			(\lambda\Id_M-\Xi)^{-1}
			\Omega(\lambda\Id_M-\Gamma).
		\end{aligned}
		\label{eq:step-Wlambda}
	\end{equation}
	Taking the $i$th column of \eqref{eq:step-Wlambda} gives
	\[
	W(\lambda)
	\psi_i
	=
	\theta\Omega^{-1}
	\begin{pmatrix}
		\displaystyle\frac{1}{\lambda-\xi_1}
		&&\\
		&\ddots&\\
		&&\displaystyle\frac{1}{\lambda-\xi_M}
	\end{pmatrix}
	\Omega
	\begin{pmatrix}
		0\\
		\vdots\\
		\lambda-\gamma_i\\
		\vdots\\
		0
	\end{pmatrix}.
	\]
	Since $\gamma_i\neq\xi_j$ for all $j$, the resolvent is regular at
	$\lambda=\gamma_i$, whereas the rightmost vector vanishes at this
	point. It follows that
	\[
	W(\gamma_i)\psi_i=0,
	\qquad
	\widehat P(\gamma_i)\psi_i
	=
	q_\xi(\gamma_i)W(\gamma_i)\psi_i
	=
	0.
	\]
	Thus $\widehat P(\lambda)$ satisfies the second set of conditions in
	\eqref{eq:spectral-interpolation-conditions} as well.
	
	Consequently, $\widehat P(\lambda)$ is a monic $2\times2$ matrix
	polynomial of degree $M$ satisfying exactly the same interpolation
	conditions as $P(\lambda)$. By the uniqueness of the interpolation
	problem, we conclude that
	$
	\widehat P(\lambda)=P(\lambda),
$
	and hence
$
	P(\lambda)=T(\lambda)=q_\xi(\lambda)W(\lambda)
$.
\end{proof}

The relation \eqref{eq:qw-gram-dressing} identifies the quasi-Wronskian and quasi-Grammian constructions at the level of their generating matrices. We next translate this correspondence into the associated ASDYM fields.

\begin{corollary}
	\label{cor:qw-gram-equivalence}
	Let the quasi-Wronskian and quasi-Grammian constructions be generated from the same spectral data. Then their Yang and potential matrices are related by
	\begin{equation}
		\mathcal J_{\mathrm W}
		=
		q_\xi(0)\mathcal J_{\mathrm G}
		=
		(-1)^M(\det\Xi)\mathcal J_{\mathrm G},
		\qquad
		\mathcal K_{\mathrm W}
		=
		\mathcal K_{\mathrm G}
		-
		(\tr\Xi)\Id_2.
		\label{eq:qw-gram-JK-relation}
	\end{equation}
	Consequently, the two constructions determine identical ASDYM gauge potentials through the Miura relations \eqref{eq:gauge-potential}, and hence represent the same ASDYM field.
\end{corollary}

\begin{proof}
By Theorem~\ref{thm:qw-gram-dressing}, we have 
$
T(\lambda)
=
q_\xi(\lambda)W(\lambda).
$
As $\lambda\to\infty$, one obtains
\[
q_\xi(\lambda)
=
\lambda^M
-
(\tr\Xi)\lambda^{M-1}
+
O(\lambda^{M-2}),
\]
and
\[
W(\lambda)
=
\Id_2
+
\lambda^{-1}\mathcal K_{\mathrm G}
+
O(\lambda^{-2}).
\]
Substituting these expansions into the relation above yields
\begin{align*}
	T(\lambda)
	&=
	\lambda^M
	\left(
	1-\lambda^{-1}\tr\Xi+O(\lambda^{-2})
	\right)
	\left(
	\Id_2+\lambda^{-1}\mathcal K_{\mathrm G}
	+O(\lambda^{-2})
	\right)
	\\
	&=
	\lambda^M\Id_2
	+
	\lambda^{M-1}
	\left(
	\mathcal K_{\mathrm G}
	-
	(\tr\Xi)\Id_2
	\right)
	+
	O(\lambda^{M-2}).
\end{align*}
Comparing the coefficient of $\lambda^{M-1}$ with that in the Darboux matrix polynomial gives
\[
\mathcal K_{\mathrm W}
=
\mathcal K_{\mathrm G}
-
(\tr\Xi)\Id_2.
\]
	
On the other hand, evaluating \eqref{eq:qw-gram-dressing}
at $\lambda=0$ gives
\[
\mathcal J_{\mathrm W}
=
q_\xi(0)\mathcal J_{\mathrm G}
=
(-1)^M(\det\Xi)\mathcal J_{\mathrm G}.
\]
Together with the relation for the potential matrices, this proves
\eqref{eq:qw-gram-JK-relation}.

Since $\tr\Xi$ and $\det\Xi$ depend only on the spectral parameters, they are independent of the spacetime variables. It follows that
\[
-(\partial_z\mathcal J_{\mathrm W})
\mathcal J_{\mathrm W}^{-1}
=
-(\partial_z\mathcal J_{\mathrm G})
\mathcal J_{\mathrm G}^{-1},
\qquad
-(\partial_w\mathcal J_{\mathrm W})
\mathcal J_{\mathrm W}^{-1}
=
-(\partial_w\mathcal J_{\mathrm G})
\mathcal J_{\mathrm G}^{-1},
\]
while
\[
\partial_{\widetilde w}\mathcal K_{\mathrm W}
=
\partial_{\widetilde w}\mathcal K_{\mathrm G},
\qquad
\partial_{\widetilde z}\mathcal K_{\mathrm W}
=
\partial_{\widetilde z}\mathcal K_{\mathrm G}.
\]
Therefore, the quasi-Wronskian and quasi-Grammian constructions determine the same gauge potentials through \eqref{eq:gauge-potential}, and hence the same ASDYM gauge field.
\end{proof}

\begin{remark}
	Theorem~\ref{thm:qw-gram-dressing} yields an
	$\SL_2(\mathbb C)$ normalization. Indeed,
	\[
	\det W(\lambda)=\frac{q_\gamma(\lambda)}{q_\xi(\lambda)},
	\qquad
	\det\mathcal J_{\mathrm G}=\frac{\det\Gamma}{\det\Xi},
	\qquad
	\tr\mathcal K_{\mathrm G}=\tr\Xi-\tr\Gamma.
	\]
	Thus, with a fixed choice of the constant square root, the matrices
	\[
	\widehat{\mathcal J}
	:=
	\left(\frac{\det\Xi}{\det\Gamma}\right)^{1/2}
	\mathcal J_{\mathrm G},
	\qquad
	\widehat{\mathcal K}
	:=
	\mathcal K_{\mathrm G}
	-\frac{\tr\Xi-\tr\Gamma}{2}\Id_2
	\]
	satisfy $\det\widehat{\mathcal J}=1$ and
	$\tr\widehat{\mathcal K}=0$, while leaving the gauge potentials
	unchanged. In particular, the associated connection is
	$\mathfrak{sl}_2(\mathbb C)$-valued.
\end{remark}

\section{Concluding remarks}
\label{sec-5}

The main result of this paper is that the quasi-Wronskian and quasi-Grammian constructions are governed by the same matrix interpolation problem. For a fixed set of pairwise distinct nonzero spectral points and the associated vacuum wave functions, the monic Darboux matrix polynomial $T(\lambda)$ is uniquely determined by the corresponding interpolation conditions. On the quasi-Grammian side, the dressing matrix $W(\lambda)$ obtained from the Cauchy matrix approach yields the same polynomial after multiplication by the scalar spectral factor $q_\xi(\lambda)$:
\[
T(\lambda)=q_\xi(\lambda)W(\lambda).
\]
This identity provides the direct algebraic link between the two constructions and shows that the quasi-Wronskian and quasi-Grammian formulations are two realizations of the same underlying spectral data.

This correspondence extends directly to their associated ASDYM fields. The quasi-Wronskian and quasi-Grammian $\tau$-functions differ only by a nonzero factor and therefore have the same nonvanishing domain. Moreover, the Yang and potential matrices are related by
\[
\mathcal J_{\mathrm W}
=
(-1)^M(\det\Xi)\mathcal J_{\mathrm G},
\qquad
\mathcal K_{\mathrm W}
=
\mathcal K_{\mathrm G}
-
(\tr\Xi)\Id_2.
\]
The differences are thus entirely contained in constant normalization factors. Consequently, both constructions produce the same gauge potentials and hence the same ASDYM gauge field.
These results clarify the relation between two solution constructions that arise from rather different formalisms. The matrix interpolation polynomial provides a common algebraic object through which their equivalence becomes transparent. 

Beyond the quasi-Wronskian and quasi-Grammian constructions considered here, the ASDYM equations admit many other classes of exact solutions obtained by different methods. A natural direction for future work is to investigate whether similar exact correspondences can be established among these different solution representations. Since many lower-dimensional integrable equations arise as symmetry reductions of the ASDYM equations, it is also interesting to examine how such correspondences behave under reduction and whether they induce analogous equivalences between different solution constructions of the reduced equations.

\section*{Acknowledgments}

The author thanks Masashi Hamanaka and Da-jun Zhang
for valuable comments, and Ken-ichi Maruno for financial
support that enabled the author to attend and give a talk
at the conference \emph{From Nonlinear Waves to Integrable Systems 2026} in Matsue, Japan.

This work was supported by the Japan Society for the
Promotion of Science (JSPS) under Grant No.~P25326
and the National Natural Science Foundation of China
(NSFC) under Grant No.~12601462.

\section*{Statements and Declarations}

\paragraph{Competing interests.}
The author declares no financial or non-financial competing
interests relevant to the content of this article.

\paragraph{Data availability.}
Data sharing is not applicable to this article, as no datasets
were generated or analysed during the current study.

\end{document}